\documentclass[sigconf, 10pt, nonacm]{acmart}
\makeatletter
\def\@ACM@checkaffil{% Only warnings
    \if@ACM@instpresent\else
    \ClassWarningNoLine{\@classname}{No institution present for an affiliation}%
    \fi
    \if@ACM@citypresent\else
    \ClassWarningNoLine{\@classname}{No city present for an affiliation}%
    \fi
    \if@ACM@countrypresent\else
    \ClassWarningNoLine{\@classname}{No country present for an affiliation}%
    \fi
}
\makeatother

\usepackage{xspace}

\newcommand{\name}{\textsc{ReTri}\xspace}

\newtheorem{lemma}{Lemma}

\author{Anton Juerss}
\affiliation{
  \institution{Weizenbaum Institute \& TU Berlin}
}

\author{Stefan Schmid}
\affiliation{
  \institution{TU Berlin \& Weizenbaum Institute}
}

\usepackage{algorithm}
\usepackage{algorithmic}
\usepackage{subcaption}

\title{Revisiting Bruck: Phase-Efficient All-to-All Collective Communication in Reconfigurable Networks}

\begin{document}

\begin{abstract}
All-to-All collective communication is a key performance bottleneck for distributed machine learning (ML) and high-performance computing (HPC) workloads, where dense traffic increasingly stresses scale-up interconnects. While these ML and HPC workloads have driven unprecedented infrastructure demand, optical reconfigurable networks (ORNs) offer a promising path forward as they can reconfigure the network at runtime. By adapting the physical topology to the active workload, they improve communication cost and bandwidth utilization. However, optical reconfigurable networks introduce a fundamental trade-off for collective communication: each reconfiguration requires global synchronization, during which communication is suspended for at a non-negligible delay. Additionally, their benefit is critically contingent on whether the collective consists of structured phases that can be served by sparse and reusable topology states. 

In this paper, we revisit Bruck's All-to-All implementation and demonstrate the benefits of topology optimization in which both communication pattern and reconfiguration strategy are co-designed. We present \name, a bidirectional All-to-All schedule for ORNs based on the Trivance algorithm. \name uses balanced ternary block propagation to complete All-to-All in $\lceil \log_3 n\rceil$ phases. The reconfiguration strategy induced by \name's pairwise bidirectional exchanges allows reconfiguration delays to be amortized across multiple phases. Preliminary simulations show that \name improves completion time by up to $10\times$ over Pairwise All-to-All, even for millisecond-scale reconfiguration delays, and improves reconfigurable Bruck by up to $2.1\times$.
\end{abstract}

\maketitle

\section{Introduction}

To meet the increasing compute and memory demands of deep learning workloads, including recommendation systems and Mixture of Expert models, modern distributed systems connect thousands of accelerators in hyperscale datacenters~\cite{Jouppi2023TPUv4,Liao2025,Qian2024Alibaba, Weiyang2023TopoOpt}. In these large-scale ML systems, efficient communication across accelerators is crucial for training and inference performance, since activations, embeddings, and tokens must be synchronized~\cite{Liao2025}. These synchronizations are typically realized through All-to-All collective communication, where each accelerator sends distinct data to every other accelerator~\cite{Thakur2005}. Their impact on end-to-end performance is substantial, accounting for up to 55\% of MoE end-to-end training time~\cite{Liao2025}. The dense communication pattern of All-to-All makes it a major performance bottleneck and increasingly stresses scale-up interconnects~\cite{Khani2021SiP}: every accelerator exchanges distinct data, raising serious concerns about congestion and requiring high network bandwidth~\cite{Liao2025,Qian2024Alibaba}. 

The design of datacenter interconnect fabrics therefore plays a central role in the scalability and efficiency of large-scale ML and HPC systems~\cite{Qian2024Alibaba, Jouppi2023TPUv4, Ding2025Rails}. Whereas conventional, electrically switched networks are power-intensive and may lead to performance bottlenecks, optical reconfigurable networks (ORNs) have emerged as a promising alternative. ORNs establish bidirectional high-bandwidth optical links between endpoints, introducing the enhanced capability to adjust and optimize the physical topology~\cite{Avin2019DAN, Griner2021, Avin2025}. However, direct optical connectivity incurs non-negligible reconfiguration delay~\cite{Porter2013,Ghobadi2016, Ding2025Rails}. Whether reconfiguration is beneficial depends on balancing its delay against the resulting reductions in path length, congestion, and transmission time. Determining when to reconfigure is itself non-trivial: per-phase reconfiguration may incur excessive overhead, whereas a static topology can forfeit substantial communication gains~\cite{Vamsi2025Bend,Weiyang2023TopoOpt}. This motivates collective schedules whose communication phases admit a small number of efficient and reusable topology states. In state-of-the-art GPU training systems, All-to-All is typically treated within scale-up domains as a destination-oriented redistribution primitive, for example in MoE token dispatch~\cite{Griner2021, Qin2025, Liao2025}. Each node partitions its payload by destination and the communication layer sends each resulting block directly to its target endpoint. If the full traffic matrix is injected as one bulk operation in an ORN, the optical layer processes dense and unstructured traffic, withholding the opportunity to adapt the topology to a specific part of the workload~\cite{Liao2025,Weiyang2023TopoOpt}.
When the transmissions of pairwise All-to-All are executed sequentially, the optical topology can in principle be reconfigured to match the traffic of each phase. However, the algorithm requires $n-1$ distinct phases and therefore a linear number of topology changes, each benefiting only a single phase and offering no structure that can be reused. Since the fixed reconfiguration delay cannot be amortized across phases, pairwise All-to-All is impractical for reconfiguration-aware execution except when reconfiguration overhead is close to zero.

This paper is driven by three simple but powerful observations. \emph{First}, even an optimal reconfiguration schedule cannot improve over a static topology if the collective algorithm does not match the constraints of ORNs. This motivates co-designing the All-to-All pattern around the port constraints of optical switches, such that communication is exposed as sparse phases. Multi-hop algorithms such as Bruck's algorithm~\cite{Bruck94} decompose All-to-All into distinct, synchronized \textit{communication phases} which are well-suited for ORNs. \emph{Second}, optimization should account not merely for bandwidth utilization, but also to reduce the number of communication phases each of which require topology adjustments for optimal communication. Fewer structured phases reduce reconfiguration requirements, while still allowing the topology to be matched efficiently to the active traffic. \emph{Third}, to fully exploit the potential of ORNs, bidirectional optical links should carry bidirectional traffic: if node $u$ sends data to node $v$, then node $v$ should simultaneously send data to node $u$ over the same optical circuit.

Based on these observations, we present a novel All-to-All communication pattern and reconfiguration strategy that addresses the described challenges: we propose \name, a bidirectional, sparse All-to-All communication pattern that synchronizes multiple pairwise exchanges to fully utilize each optical link. The induced reconfiguration schedule forms connected subrings that serve the active traffic while preserving persistent topology states for subsequent phases. Leveraging Trivance’s communication pattern~\cite{juerss2026}, blocks are routed over shortest paths with more blocks per transmission. Therefore, \name reduces the number of communication phases and thereby the associated number of reconfigurations. With two electrical-to-optical transceivers per accelerator, \name completes All-to-All in $\lceil \log_3 n\rceil$ phases, an effective reduction by $33\%$ compared to Bruck's All-to-All ORN-feasible implementation. Our preliminary results show that \name reduces All-to-All completion time by up to $90\%$ for small reconfiguration delays of $1\,\mu\mathrm{s}$ and consistently outperforms Bruck by up to $2.1\times$. Even for large reconfiguration delays, \name achieves speedups of $1.5\times$ to $6.9\times$ for $1\,\mathrm{ms}$ and showing improvements even for $50\,\mathrm{ms}$ delay for large workloads. 
 
\section{Motivation: Why All-to-All Changes Under Reconfiguration}

Prior work has established that collective communication is particularly compatible with optical reconfigurable networks (ORNs) since the communication pattern is known in advance, and collectives proceed in synchronized phases, between which the network can be reconfigured~\cite{ Vamsi2025Bend, Kumar2024Case}. Existing \textbf{phased All-to-All algorithms}, including Bruck's algorithm, are typically designed for logical one-directional communication patterns and therefore, by design, do not fully exploit bidirectional optical links~\cite{Bruck94,Thakur2005}. A common approach to improve network utilization --- used for ring-based collectives such as AllReduce --- is to \textbf{mirror} the schedule by halving the workload and executing one half in each direction, improving bandwidth utilization by up to $2\times$~\cite{Daniele2024Swing, juerss2026,Sack2015}. In ORNs, however, optimizing for bandwidth utilization alone is insufficient: since reconfiguration is costly, each optimized topology state should carry as much All-to-All traffic as possible. The objective is therefore to minimize the number of topology adjustments while still transmitting over direct, optimal optical connections.

\begin{figure}[t]
    \centering
    \includegraphics[width=\columnwidth]{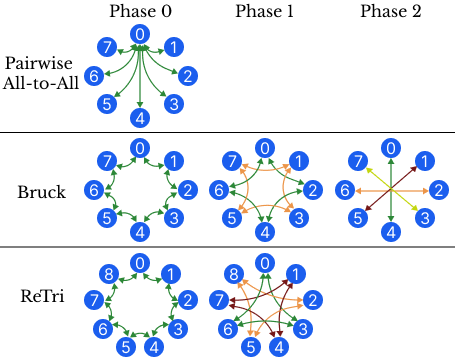}
    \caption{Communication phases of Bruck, \name, and Pairwise All-to-All (for clarity, only transmissions from node~0 are shown). \name completes in one fewer phase than Bruck with one additional node. Arrow colors identify the connected subrings used for reusable communication.}
  \label{fig:intro}
  \vspace{-10px}
\end{figure}

For our \textbf{reconfigurable architecture}, we consider a scale-up domain, i.e., a tightly coupled group of $n$ accelerator endpoints within a single server or memory domain, connected to a programmable optical interconnect with $p$ endpoint-facing optical ports. The interconnect establishes bidirectional optical circuits between endpoint ports and can reconfigure these circuits on demand, incurring a reconfiguration delay $\alpha_r$\cite{Porter2013,Farrington2013,Liao2025,Ghobadi2016}. If each node has only one transceiver, with $p=n$, the interconnect can realize only a peer-to-peer matching in each configuration: every node can be connected to at most one other node~\cite{Vamsi2025Bend, Porter2013, Mellette2017RotorNet}. Such a topology cannot be connected for $n>2$, and All-to-All must therefore be served by repeatedly reconfiguring direct pairings across phases. To maintain a persistent connected topology over all endpoints, each node must have a degree of at least two~\cite{Ding2025Rails, juerss2026bridgeoptimizingcollectivecommunication}. Under this constraint, the minimal topology is a ring, which requires two electrical-to-optical transceiver per node and hence $p=2n$ endpoint-facing ports at the optical-circuit switch (OCS)~\cite{Ding2025Rails}.

Pairwise All-to-All is a poor fit for reconfigurable networks, as this pattern misses \textbf{structured phases} for which the topology can adjust for (Figure~\ref{fig:intro}). Each phase should expose \emph{only} as many concurrent data transmissions as there are available ports --- under the minimal assumption of two transceivers per node and $2n$ optical ports at the OCS, this corresponds to $2n$ simultaneous directed transmissions. Bruck's algorithm with additional mirroring splits All-to-All in $\lceil \log_2 n\rceil$ phases, which together form connected subrings of size $\frac{n}{2^k}$ in phase $k$ to ensure persistent topology states. These patterns are visualized in Figure~\ref{fig:intro}. This raises the question: \textit{Can we achieve the same property with fewer phases?} Theoretically, with two transceivers per node, All-to-All can be completed in $\lceil \log_3 n \rceil$ phases.

Assuming an initial logical ring, Bruck's two-port pattern makes each node communicate with peers at offsets $3^k$ and $2\cdot 3^k$. This achieves $\lceil \log_3 n\rceil$ phases in an ideal multi-port message-passing model, but it is not directly feasible for ORNs. Given the OCS has $2n$ ports, it can only establish $n$ bidirectional matching optical links, i.e., two incident links per node. Bruck's pattern, however, uses established links only in one direction. This generally requires each node to send to two peers and receive from two different peers in the same phase; realizing these non-pairwise exchanges would require up to four incident optical links per node. Moreover, when Bruck is implemented over a ring, its traffic still progresses along one logical direction, increasing congestion and propagation delay. 

Thus, the \textbf{effectiveness of reconfiguration schedules} in ORNs is determined by the underlying All-to-All communication pattern. This opens the design space for an All-to-All schedule that completes in $\lceil \log_3 n\rceil$ phases. Such schedule must be bidirectional by design, using both directions of each optical link while preserving reusable connected topologies and balancing link utilization. In the next section, we present \name, an algorithm that addresses these requirements and explores the trade off between performance gains and reconfiguration overhead.

\section{\name: Shortest Path All-to-All}
\label{sec:retri}

We present \name, an efficient bidirectional All-to-All communication pattern tailored for ORNs, that establishes reusable subrings for sparse reconfiguration for $2n$ OCS ports completing in $\lceil \log_3 n \rceil$ phases.

\subsection{Ternary Communication Pattern}
In order to make optimal use of both available ports at each node for bidirectional communication, we use the Trivance approach by Juerss et al.~\cite{juerss2026}. Rather than forwarding data along a single logical direction, each node communicates in Trivance with two symmetric peers in every phase. Specifically, in phase $k$, each node $r$ in an $n$-node network communicates with two peers defined as:
\begin{center}
    $\pi(r, k, n) = (\pi_{\text{left}},\pi_{\text{right}}) = 
\begin{cases}
\pi_{\text{left}} = r - 3^k \bmod n,\\
\pi_{\text{right}} = r + 3^k \bmod n,
\end{cases}
$
\end{center}
In AllReduce, Trivance exposes a latency--bandwidth tradeoff: it reduces the number of communication phases---matching the lower bound and shortens the paths lengths compared to a ring-based AllReduce. This improves completion time for small to medium message sizes. All-to-All exhibits a different set of constraints. Since transmitted blocks are distinct and cannot be reduced, every additional forwarding hop contributes directly to propagation delay, link occupancy, and congestion. By contrast, the bidirectional pattern underlying Trivance forwards blocks in both logical directions and thereby shortcuts the ring. When applied to All-to-All, this structure reduces the number of communication phases and decreases path lengths. As in Bruck's two-port version, each node must receive both incoming transmissions of the current phase before it can advance to the next. Consequently, \name improves effective throughput as well as propagation delay compared to Bruck. Mirroring Bruck's pattern in the opposite direction can improve port utilization, but it is still constrained to $\lceil \log_2 n \rceil$ phases.

\subsection{Block Propagation in All-to-All}
The ternary pattern of Trivance reaches all nodes within $s = \lceil \log_3 n \rceil$ communication phases. The resulting data movement, however, must be defined separately from AllReduce. In AllReduce, data blocks can be reduced between phases at each node, whereas in All-to-All every transmitted block is unique. In each phase of \name, the ternary structure partitions the remaining data into three groups: blocks that stay local, blocks propagated to the left, and blocks propagated to the right. Thus, for an initial message of $m$ bytes per node, each node transmits $\frac{m}{3}$ bytes in each direction. Since every phase expands the set of reachable destinations by a factor of three, the canonical network size for \name is $n=3^s$.

We define the block propagation of \name analogously to Bruck's radix-$3$ All-to-All~\cite{Bruck94}, but using balanced ternary digits. We assume that $n=3^s$ and denote $B[r,d]$ as the block initially stored at node $r$ and destined for node $d$. We define its
signed offset as
$
\Delta_{r,d}=\operatorname{ucr}_n((d-r)\bmod n)
\in \{\frac{-(n-1)}{2},\dots,0,\dots,\frac{n-1}{2}\},
$
where $\operatorname{ucr}_n(\cdot)$ maps a distance offset modulo $n$ to its unique centered representative. Since $n=3^s$, every offset has a unique
balanced ternary representation
$
\Delta_{r,d}=\sum_{k=0}^{s-1} \tau_k(r,d)3^k$ with $\tau_k(r,d)\in\{-1,0,+1\}.$
This representation defines how each block reaches its destination. In communication phase $k$, node $i$ sends all currently stored
blocks with $\tau_k=+1$ to $(i+3^k)\bmod n$, with $\tau_k=-1$ to $(i-3^k)\bmod n$. Blocks with
$\tau_k=0$ are not sent in phase $k$. Hence, the two outgoing messages
of node $i$ in phase $k$ are
$
M_{i,+}^{(k)}
=
\{B[r,d] : i = (r+\sum_{\ell<k}\tau_\ell(r,d)3^\ell)\bmod n,
\ \tau_k(r,d)=+1\},
$ send to $(i+3^k)\bmod n$
and
$
M_{i,-}^{(k)}
=
\{B[r,d] : i = (r+\sum_{\ell<k}\tau_\ell(r,d)3^\ell)\bmod n,
\ \tau_k(r,d)=-1\}.
$ send to $(i-3^k)\bmod n$. After all $s=\log_3 n$
phases, block $B[r,d]$ has moved by
$
\sum_{k=0}^{s-1}\tau_k(r,d)3^k=\Delta_{r,d}
$
positions and therefore reaches its destination $d$. The full correctness proof can be found in Appendix~\ref{appendix_blocks}. For example, for $n=9$, a block whose destination is two positions to the left has offset $-2$. This offset can be represented by the balanced ternary digits $(1,-1)$, since the block first moves one phase to the right in phase $0$ and $3$ to the left in phase $1$.

\subsection{Reconfiguration Strategy: Reusable Ternary Subrings}
%To enable persistent reconfiguration, the network should maintain, under the given port limitation, a connected topology to future peers across multiple collective phases instead of requiring reconfigurations at every phase. Ring-like fabrics or rail structured scale-out systems are therefore a practical and common design for reconfigurable distributed ML systems. Rings emerge as the lowest-degree connected topology that preserves reachability among all nodes. As shown in prior work, reconfigurations that create connected subrings can significantly reduce communication cost by shortcutting the network with optical links, while maintaining reachability to future communication peers. This allows a single topology reconfiguration to benefit multiple subsequent phases without requiring further topology adjustments. We build on this observation and study how All-to-All communication can further exploit such persistent, ring-based reconfiguration patterns.

The performance on static rings of \name's ternary communication pattern is similar to the performance of shortest-path source-destination All-to-All which only requires a single communication phase. However, its main advantage is the algebraic structure of its phases from which we can directly derive our new topology states for each reconfiguration. This structure makes the active communication pattern compatible with sparse, reusable optical topologies under the two-port constraint. Given the OCS has $2n$ ports, each node can establish exactly optical connections to two other nodes at any time. Hence, a feasible topology is a degree-two bidirectional graph: every node has two incident optical links, and the resulting topology is a collection of rings~\cite{Ding2025Rails}. For \name, these rings are induced directly by the ternary communication pattern. Given a reconfiguration schedule $\mathbf{x}=(x_0,\ldots,x_{s-1})$, where $x_k=1$ denotes a reconfiguration before phase $k$, the ORN configures the edge set
$E_k=\{\{i,(i+3^k)\bmod n\}: i\in\{0,\ldots,n-1\}\}$.
Equivalently, each node $i$ is connected to the two peers $(i-3^k)\bmod n$ and $(i+3^k)\bmod n$. Algorithm~\ref{alg:retri} defines the edge-set construction of \name and the corresponding subrings induced per phase.

\begin{algorithm}[t]
  \caption{\name: Reusable Ternary Subrings}
  \label{alg:retri}
  \begin{algorithmic}[1]
    \REQUIRE $n=3^s$ nodes, reconf. schedule $\mathbf{x}=(x_0,\ldots,x_{s-1})$
    \STATE $s\gets \log_3 n$
    \FOR{each phase $k\in\{0,\ldots,s-1\}$}
        \IF{$x_k=1$}
            \FOR{each residue $i\in\{0,\ldots,3^k-1\}$}
                \STATE Construct subring $S_i^{(k)}=\{u\mid u\equiv i \pmod{3^k}\}$
                \FOR{each node $u\in S_i^{(k)}$}
                    \STATE Set bidirectional optical links to $u-3^k \bmod n \quad\text{and}\quad u+3^k \bmod n$
                \ENDFOR
            \ENDFOR
        %\ELSE
            %\STATE Reuse the topology from phase $k-1$
        \ENDIF
    \ENDFOR
  \end{algorithmic}
\end{algorithm}

While this reconfiguration schedule directly connects peers for the next phase, it also creates a new subring illustrated in Figure~\ref{fig:intro}. For a reconfiguration before phase $k$, the topology is partitioned into $3^k$ subrings. Each subring has size $\frac{n}{3^k}$. For each residue class $i\in\{0,\ldots,3^k-1\}$, define
\[
    S_i^{(k)}
    :=
    \{\,u\in\{0,\ldots,n-1\}\mid u\equiv i \pmod{3^k}\,\}.
\]
Prior work showed that for OCSes with $2n$ ports, reconfiguration strategies that induce subrings reduce communication cost of the active phase while preserving reachability for subsequent phases~\cite{juerss2026bridgeoptimizingcollectivecommunication, rahman2026harvestadaptivephotonicswitching}. We prove in Lemma~\ref{lemma:subrings} in the Appendix~\ref{appendix_subrings} that the subrings induced by \name are minimal and include all nodes of future peers.

\subsection{Performance Gains and Reconfiguration Costs}
Reconfiguration during collective operations introduces a fundamental tradeoff: reconfigurations can adapt the topology to the next communication phase and thereby reduce completion time, but each adjustment incurs a reconfiguration overhead $\delta$. We analyze this tradeoff for \name using an extended Hockney $\alpha$-$\beta$ cost model, following related work~\cite{Won2023Astrasim, Vamsi2025Bend, juerss2026, Aashaka2023TACCL}:
\begin{equation*}
    C^A(m) = \underbrace{s \cdot \alpha_s}_{\text{per-phase delay}}
    + \sum_{k=0}^{s-1} (\underbrace{h_k \cdot \alpha_h}_{\substack{\text{per-hop}\\\text{delay}}}
    + \underbrace{m_{k} \cdot c_{k} \cdot \beta}_{\substack{\text{transmission}\\\text{delay}}})
    + \underbrace{R\cdot \delta}_{\substack{\text{reconf.}\\\text{delay}}}
\end{equation*}
where, in each phase $k$ of total $s$ phases, the algorithm $A$ incurs a startup latency $\alpha_s$ (e.g., data preparation), per-hop delay $\alpha_h$ for each hop $h_k$, a transmission delay of $\beta \cdot m_k \cdot c_k$ with $m_k$ as the chunk size transmitted in phase $k$, $\beta=\frac{1}{b}$ the network cost per byte based on bandwidth $b$, $c_k$ maximum network congestion per directional link in phase $k$ and a reconfiguration overhead $\delta$ for the number of reconfigurations $R$. The total cost of \name to complete All-to-All in a static ring network of $n$ is as follows:
\[
\begin{aligned}
C^{\name}(m)
&=
\underbrace{\log_3 n \cdot \alpha_s}_{\text{per-phase delay}}
+
\sum_{k=0}^{\log_3 n - 1} (\underbrace{3^k \cdot \alpha_h}_{\substack{\text{per-hop}\\\text{delay}}} + \underbrace{ \frac{m\cdot3^k\cdot\beta}{3}}_{\substack{\text{transmission}\\\text{delay}}})\\
&=
\log_3 n \cdot \alpha_s
+
\left(\alpha_h+\beta\frac{m}{3}\right)
\cdot
\frac{n-1}{2}.
\end{aligned}
\]
Let $\mathbf{x}=(x_0,\ldots,x_{s-1})\in\{0,1\}^s$ denote a reconfiguration schedule, where $x_k=1$ indicates that the OCS is reconfigured before phase $k$, and $x_k=0$ that the previous topology is reused. If a topology is configured at phase $k$, then phase $k+t$ can be served over the same subrings with hop distance $3^t$. Hence, a phase segment of length $r$, where the topology is not reconfigured, has communication cost 
\[
    C_{\mathrm{seg}}^{\name}(r)
    =
    \sum_{t=0}^{r-1}
    \left(
        \alpha_s
        +
        \left(\alpha_h+\beta\frac{m}{3}\right)3^t
    \right)
    =
    r\alpha_s
    +
    y\frac{3^r-1}{2},
\]
where $y := \alpha_h+\beta\frac{m}{3}$ and each node transmits two message of size $\frac{m}{3}$ in each direction. After a topology reconfiguration, the communication costs of the next phase is minimal: each node is directly connected to both its peers. Each additional phase without reconfiguration increases the communication distance, and therefore the congestion and propagation delay by a factor of three. Thus, the placement of reconfigurations determines how the $s=\log_3 n$ communication phases are partitioned into topology segments. Following prior work, the optimal schedule balances these segment lengths~\cite{juerss2026bridgeoptimizingcollectivecommunication}: for a fixed number $R$ of reconfigurations, the resulting $R+1$ segments should differ in length by at most one~\cite{juerss2026bridgeoptimizingcollectivecommunication}. In the aligned case where $R+1$ divides $\log_3 n$, all segments have length $\frac{\log_3 n}{R+1}$, and the cost for R reconfigurations is
\[
    C^{\name}(R)
    =
   \log_3 n \cdot \alpha_s
    +
    (R+1)\cdot (\alpha_h+\beta\frac{m}{3})\cdot
    \frac{3^{\frac{\log_3 n}{R+1}}-1}{2}
    +
    R\delta .
\]
If the ORN is reconfigured before every phase with $R = s -1$ and therefore optimize the topology between each phase reducing the distance of communicating peers to 1. This leads to the optimal costs: 
\[
\begin{aligned}
C^{\name}(\log_3 n -1)
&=
\log_3 n \left(\alpha_s+\alpha_h+\beta\frac{m}{3}\right)
+
(\log_3 n -1)\delta .
\end{aligned}
\]

\noindent Given this, frequent reconfiguration between phases improves the performance and reduces the communication cost of All-to-All by $\left(\alpha_h+\beta\frac{m}{3}\right)\left(\frac{n-1}{2}-\log_3 n\right)$ at the cost of additional reconfiguration delay of $(\log_3 n-1)\delta$.
In comparison, Bruck requires $\log_2 n$ communication phases between which the topology may reconfigure, leading to maximal performance gains:
\[
\begin{aligned}
C^{Bruck}(\log_2 n -1) =
\log_2 n \left(\alpha_s+\alpha_h+\beta\frac{m}{4}\right)
+
(\log_2 n -1)\delta .
\end{aligned}
\]
Bruck requires $\frac{\log_2 n}{\log_3 n} = \log_2 3 \approx 1.58\times$ as many phases as \name. While transmission delay is for every-phase reconfiguration identical, Bruck incurs about $58\%$ higher per-phase latency and per-hop latency. More importantly, achieving the same topology optimization inflicts also $58\%$ more reconfiguration delay $\delta$. For large ORNs systems with millisecond-scale reconfiguration delay, this difference can determine whether reconfigurations are even beneficial which gives \name with a significant advantage over Bruck.

\noindent We analyzed the cost-optimal schedules for a fixed number of reconfigurations. This raises the questions of \textit{when and how often are reconfigurations beneficial for \name?} Given the network parameters, the answer is obtained by evaluating the completion time for each feasible number of reconfigurations and selecting the lowest cost schedule. For \name, this corresponds to $R^\star=\arg\min_{0\le R\le s-1} C^{\name}(R)$, where $s=\log_3 n$. The optimal value of $R$ depends on the relative cost of communication and reconfiguration. Reconfiguration between each phase is beneficial when the optimization in hop distance and congestion exceeds the reconfiguration overhead. This is more likely for large $n$, large message size $m$, low bandwidth, or high per-hop delay.
\section{Preliminary Evaluation}

Our analysis explores a clear tradeoff between performance gains by topology adjustments and reconfiguration overhead. The central questions to determine \name's effectiveness are: under which network parameters do reconfigurations with \name improve static All-to-All completion time, and how does \name compare to existing All-to-All reconfiguration schedules? To answer these questions, we conduct preliminary simulations of \name using Astra-Sim~\cite{Won2023Astrasim} with ns-3~\cite{ns3-simulator} as the network backend. We compare \name against two baselines. The first is the static Pairwise All-to-All algorithm, which operates on a static ring. It injects all blocks for each destination concurrently, not relying additional synchronizations within the collective. Consequently, it attains the communication cost lower bound. The second baseline is the reconfigurable Bruck algorithm, namely Bridge~\cite{juerss2026bridgeoptimizingcollectivecommunication}, which established connected subrings based on Bruck's~\cite{Bruck94} communication pattern, which are reusable. A reconfiguration at step $k$ establishes topology links directly between the communicating peers, similar as in \name. We extend Bridge by employing a mirrored All-to-All with half the data $m$; thereby both logical directions are used to provide a fair comparison to \name. We model a scale-up system with $400\,\mathrm{Gbps}$ link bandwidth, a propagation delay of $1\,\mu\mathrm{s}$, and a per-phase delay of $1.7\,\mu\mathrm{s}$, following configurations used in related work~\cite{Liao2025, Aashaka2023TACCL, rahman2026harvestadaptivephotonicswitching, Farrington2013}. We evaluate workloads from $m=1\,\mathrm{KB}$ to $256\,\mathrm{MB}$ and vary the reconfiguration delay $\delta$ over a broad range, from $1\,\mu\mathrm{s}$ to $50\,\mathrm{ms}$, to cover both fast optical interconnects with limited port number and architectures that provide more ports but incur higher reconfiguration delays~\cite{Farrington2010Helios, polatis7000series, calient2022ocsdatasheet}. Since \name is naturally aligned with powers of three, while Bruck with powers of two, we compare each algorithm at its favorable size. In Figures~\ref{fig:maxtrix_64_static} and~\ref{fig:maxtrix_64_bruck}, Bruck operates on $64$ nodes, whereas \name operates on $81$ nodes, which favors Bruck --- All-to-All completion time increases linearly with network size.

\begin{figure}[t]
    \includegraphics[width=\linewidth]{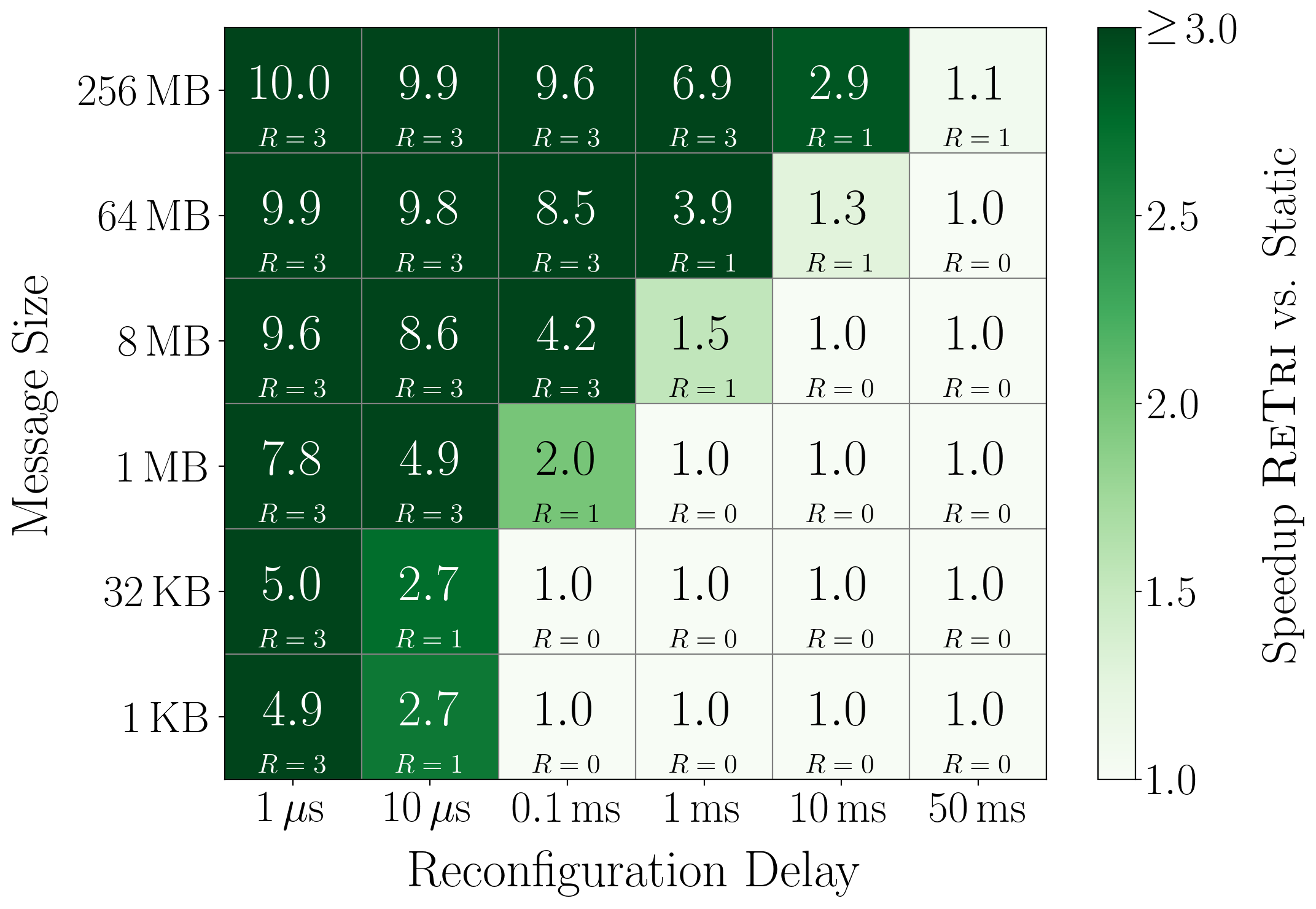}
    \caption{Heatmap showing All-to-All speedups of \name ($n=81$) compared to static All-to-All ($n=64$) with $R$ denoting the number of reconfigurations by \name. Color intensity represents the speedup, with values capped at $3.0\times$. As message size increases, \name reconfigures more often (increasing $R$), improving over Pairwise All-to-All by up to $10\times$.}
    \label{fig:maxtrix_64_static}
\end{figure}

\begin{figure}[t]
    \includegraphics[width=\linewidth]{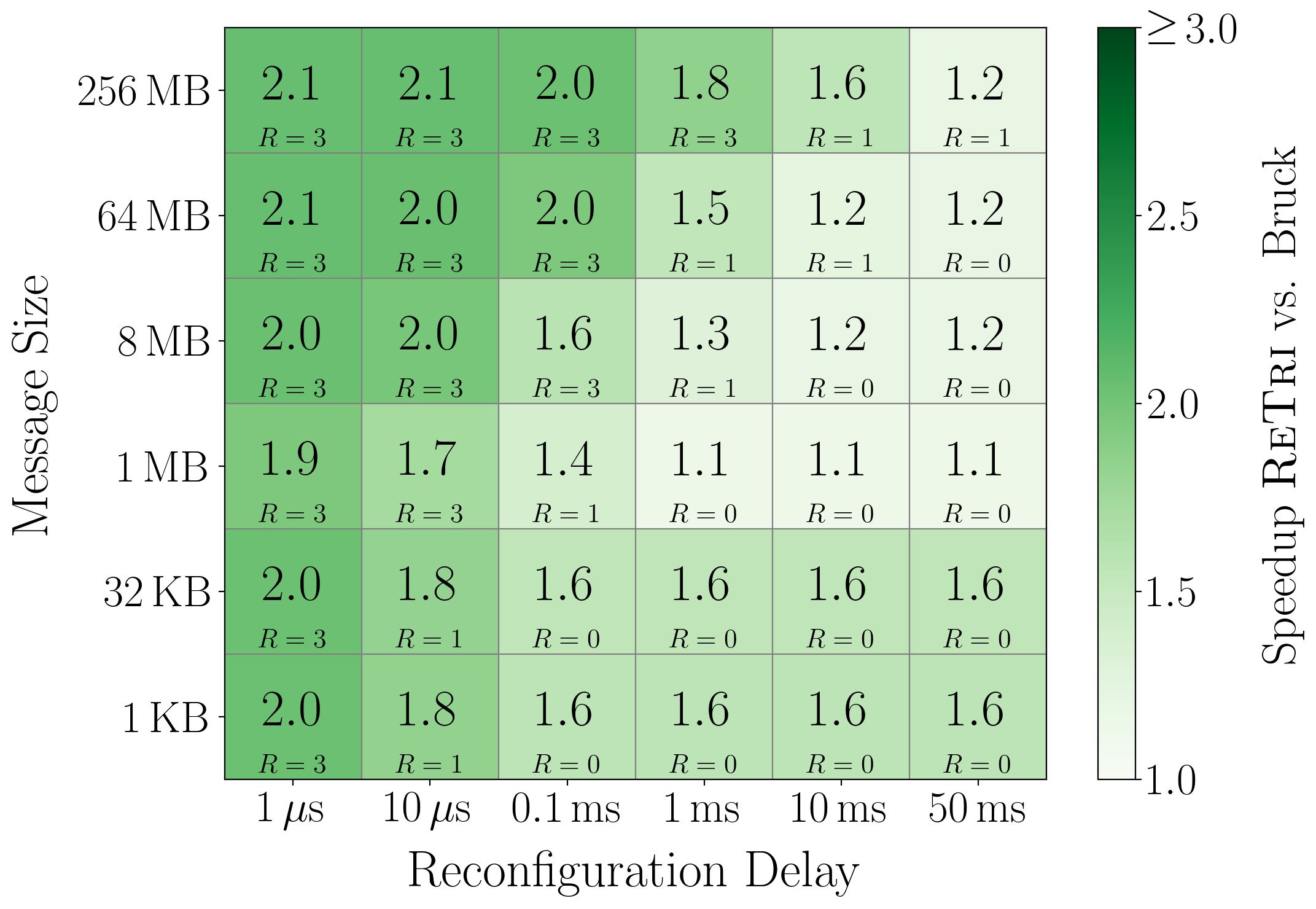}
    \caption{Heatmap showing All-to-All speedups of \name ($n=81$) compared to Bridge, Bruck's reconfiguration strategy ($n=64$). As reconfiguration delay increases, \name reconfigures less often (reducing $R$), while retaining gains from its static communication schedule when $R=0$.}
    \label{fig:maxtrix_64_bruck}
\end{figure}

Figure~\ref{fig:maxtrix_64_static} reports the speedup of \name over static shortest-path All-to-All. For low reconfiguration delay, $\delta=1\,\mu\mathrm{s}$, \name achieves speedups of up to $10\times$. In general, larger messages increase the completion time, so the proportional performance gains from reconfiguration grow and can outweigh the fixed overhead $\delta$. Reconfiguration remains beneficial up to $\delta=10\,\mu\mathrm{s}$ for small messages, up to $\delta=1\,\mathrm{ms}$ for messages up to $8\,\mathrm{MB}$, and even up to $\delta=50\,\mathrm{ms}$ for $256\,\mathrm{MB}$ messages. Compared to Bruck in Figure~\ref{fig:maxtrix_64_bruck}, \name achieves consistent speedups of up to $2.1\times$, despite operating on $81$ nodes compared to Bruck on $64$ ($26\%$ larger network size). For small messages, \name benefits primarily from its lower phase count and shorter forwarding paths, yielding consistent speedups of at least $1.6\times$. The number of reconfigurations in this regime is $R=0$; hence, \name's improvement stems from its advantage in static networks, where the number of phases and path length dominate performance for small messages. For larger messages, speedups range from $1.2\times$ to $2.1\times$, where transmission delay dominates and reconfigurations are again feasible. 

For low $\delta$, frequent topology updates yield $5$--$10\times$ speedups over the static baseline, while \name still maintains gains of up to $1.5\times$ at $8\,\mathrm{MB}$ and $1.1\times$ even with $\delta=50\,\mathrm{ms}$ at $256\,\mathrm{MB}$. These results identify the regimes in which reconfigurations are beneficial: as message size and network size increase, the performances gains from optimizing the topology increasingly outweigh the reconfiguration overhead. Figures~\ref{fig:maxtrix_8_all} and~\ref{fig:maxtrix_256_all} in the appendix further show that, for larger networks, \name remains beneficial even at high reconfiguration delays, improving over the static baseline by $1.2\times$ at $\delta=150\,\mathrm{ms}$ for $256\,\mathrm{MB}$ messages, whereas Bruck no longer improves over static execution. This confirms that \name is most effective when communication is expensive: in these regimes, sparse or frequent reconfiguration substantially reduces hop distance and link congestion. Across the evaluated parameter space, \name also provides an effective reconfiguration strategy than Bruck.

\section{Discussion, Challenges, and Research Agenda}
We establish \name as a concrete step toward reconfiguration-aware All-to-All in scale-up reconfigurable domains. We discuss the remaining challenges that shape a broader research agenda across algorithm design and systems integration.

\textbf{Non-power-of-three Networks: }Our analysis focuses on the aligned case $n=3^s$, where our ternary pattern induces nested reusable subrings. For arbitrary network sizes, \name establishes a ring under the identical pattern which performs identical to the next largest power-of-three size. Consequently, \name achieves the same performance at, for example, 64 and 81 nodes. Although All-to-All completion time increases almost linearly with network size, \name achieves significantly lower completion time than Bruck despite Bruck operating on a network that is $20\%$ smaller. Future work should optimize \name for power-of-two network sizes, which are common in TPU deployments~\cite{Jouppi2023TPUv4}, while accounting not only for path length and congestion but also for structural properties such as link reusability.

\textbf{Extension Beyond Rings: }
We can generalize \name beyond rings, with $d=2q$ optical ports per node to balanced radix $b=d+1$, where phase $k$ connects each node to the peers at offsets $\pm a b^k$. The resulting topology then consists of a degree-$d$ circular ring. We leave this for future work.

\textbf{Overlapping Computation with Reconfigurations: }
Reconfiguration delay does not necessarily lie on the critical path. In practical workloads, topology changes may overlap with computation or data preparation. For \name, this implies that the practical benefit of reconfiguration-aware schedules may be larger than suggested.

\textbf{Other Collectives: }
Although \name targets All-to-All, the same design principle may extend to other collectives with structured communication phases for ORNs, such as AllReduce. These collectives introduce additional dependencies as blocks may be reduced or replicated across phases. In particular, the ring algorithm completes AllReduce with minimal data transmission and congestion, so \name could be applied to primarily improve phase count and propagation delay.

\textbf{Routing Challenges: }
When the topology is reconfigured for the active phase, all scheduled transmissions are direct. When a topology is reused across later phases, however, blocks may traverse multiple hops inside a subring. Efficient routing must therefore preserve the balanced bidirectional structure of \name while avoiding congestion and head-of-line blocking.

\textbf{Synchronization Between Reconfigurations: }
Reconfiguration requires all nodes to complete the current phase before optical circuits can adjust. Barrier overhead, stragglers, and control-plane latency can reduce the benefit of frequent topology updates. While, we show that \name improves All-to-All for reconfiguration delays of $50\,\mathrm{ms}$ to $150\,\mathrm{ms}$, future work may explore synchronization challenges further.

\textbf{Existing reconfiguration strategies: }
Existing reconfiguration strategies for collective workloads provide important insights for adapting optical topologies, but they can overcomplicate the problem by optimizing topology changes independently of the communication pattern. In many cases, the reconfiguration schedule can be derived directly from the communication pattern and is deterministic as the workload is known for All-to-All and AllReduce; when the pattern is constrained for ORNs, however, the remaining topology optimization is limited in performance and possibilities.

In addition to the larger research avenues discussed, technical questions remain open. These include extending the evaluation of \name to hardware-level simulation with a concrete OCS model. More broadly, \name motivates a design space of reconfiguration-aware collectives in which the communication schedule and optical topology are co-designed rather than optimized independently.
\vspace{-5px}
\section{Conclusion}
By revisiting Bruck's logarithmic phase structure under physical ORN constraints, we proposed \name, an All-to-All algorithm which achieves $\lceil \log_3 n\rceil$ communication phases while preserving pairwise bidirectional exchanges and inducing reusable subring topologies. This structure exposes a direct tradeoff between communication distance and reconfiguration overhead: frequent reconfiguration reduces hop count and congestion, while sparse reconfiguration amortizes topology changes across multiple phases. We showed that reducing the number of communication phases directly improved completion time by up to $2.1\times$ to Bruck.

We believe this direction is particularly promising for future scale-up systems, where dense All-to-All traffic increasingly stresses electrical interconnects. By combining structured collective phases with reusable topologies, reconfigurable networks can move beyond direct peer matching and support communication patterns that are both phase-efficient and benefit over realistic reconfiguration delays.\\

\noindent \textbf{{\large Acknowledgments}}

\medskip
\noindent We thank Volker Stocker for his helpful inputs and feedback. This work is part of a project that has received funding from the European Research Council (ERC), project FortifyNet (grant 101287293), 2026-2027 and by the German Federal Ministry  of  Research,  Technology  and  Space  (BMFTR)  under grant 16DII141 “Weizenbaum Institut für die vernetzte Gesellschaft”.

\begin{figure}[!h]
    \centering
    \includegraphics[width=0.6\linewidth]{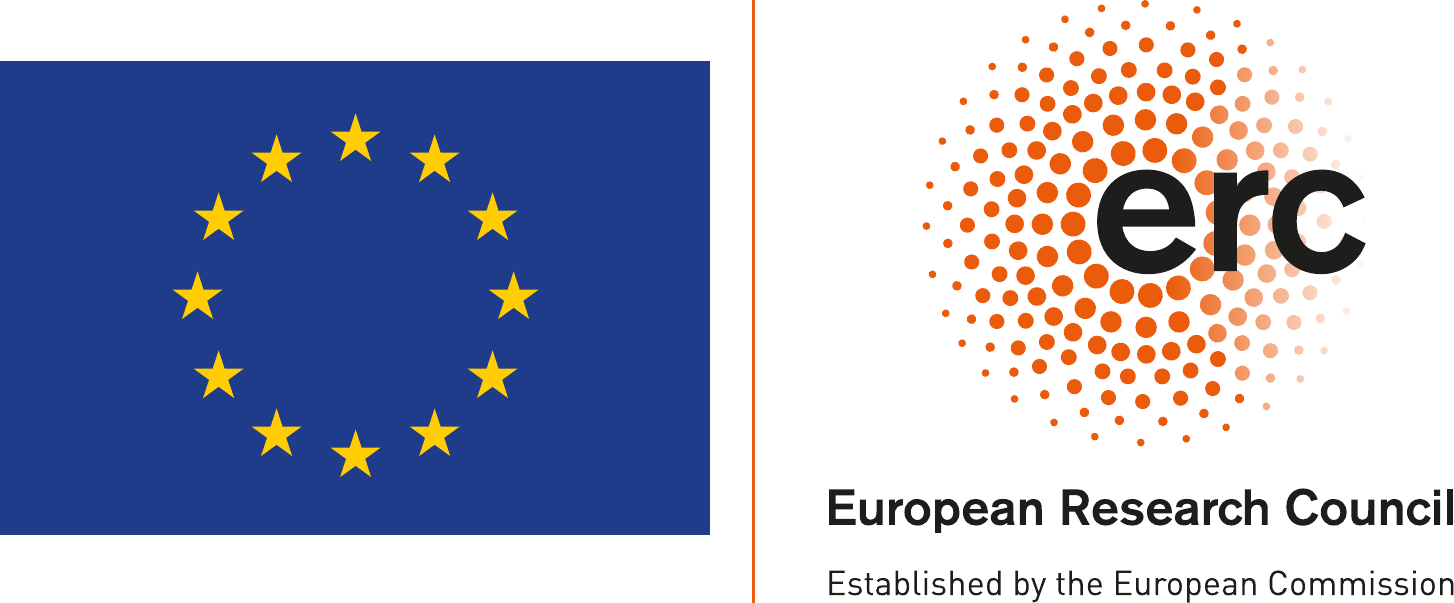}
    \label{fig:my_label}
\end{figure}

\bibliographystyle{ACM-Reference-Format}
\bibliography{sample-base}

\appendix

\section{Proof of Minimal Subrings for \name}
\label{appendix_subrings}
\begin{lemma}
\label{lemma:subrings}
For phase $k$, the subrings $S_i^{(k)}$ are minimal under the $2n$ port constraint: they contain exactly the nodes that must remain mutually reachable for phase $k$ and all later phases, and the induced subring uses the minimum degree needed to support bidirectional communication at every node.
\end{lemma}

\begin{proof}
    Node $u$ communicates in phase $k$ with $u\pm 3^k$, which lie in the same residue class modulo $3^k$. For every later phase $j>k$, the offset is $3^j=3^{j-k}\cdot 3^k$, so all future peers of $u$ also lie in the same residue class. Therefore, any reusable topology established at phase $k$ must keep all nodes in $S_i^{(k)}$ connected. Conversely, no node outside $S_i^{(k)}$ is needed for any future offset, since all future offsets are multiples of $3^k$. Finally, every node must be able to communicate bidirectionally with its two phase neighbors; this requires degree two, which is exactly realized by the subring on $S_i^{(k)}$. Hence, the subring is minimal and includes all future peers.
\end{proof}

\section{Ternary Block Propagation Proof}
\label{appendix_blocks}
The propagation rule of \name relies on $n=3^s$ and every source-destination offset has a unique signed base-$3$ representation. For a block $B[r,d]$, let $\Delta_{r,d}=\operatorname{ucr}_n((d-r)\bmod n)$ be its centered offset. We represent this offset as $\Delta_{r,d}=\sum_{k=0}^{s-1}\tau_k(r,d)3^k$, where $\tau_k(r,d)\in\{-1,0,+1\}$. The digit $\tau_k$ determines the action in phase $k$: the block is sent left if $\tau_k=-1$, kept local if $\tau_k=0$, and sent right if $\tau_k=+1$.

\begin{lemma}
\label{lem:balanced-ternary-propagation}
Assume $n=3^s$. For every source node $r$ and destination node $d$, the offset $\Delta_{r,d}$ has a unique representation $\Delta_{r,d}=\sum_{k=0}^{s-1}\tau_k(r,d)3^k$ with $\tau_k(r,d)\in\{-1,0,+1\}$. Moreover, in every phase $k$, each node sends exactly $\frac{n}{3}$ blocks to the left and exactly $\frac{n}{3}$ blocks to the right.
\end{lemma}

\begin{proof}
First, we show uniqueness. Assume for contradiction that two distinct vectors $a,b\in\{-1,0,+1\}^s$ represent the same offset, i.e.,
$\sum_{k=0}^{s-1}a_k3^k=\sum_{k=0}^{s-1}b_k3^k$.
Then $\sum_{k=0}^{s-1}(a_k-b_k)3^k=0$, where each coefficient $a_k-b_k$ lies in $\{-2,-1,0,1,2\}$. Let $j$ be the largest index with $a_j\neq b_j$. The leading term has absolute value at least $3^j$, whereas all lower-order terms summed up can have at most absolute value $2\sum_{\ell<j}3^\ell=3^j-1$. Hence the leading term cannot be substituted, a contradiction which leads to a different destination. Thus the representation is deterministic. Since there are $3^s=n$ unique signed digit vectors and exactly $n$ centered offsets in $\{-(n-1)/2,\ldots,(n-1)/2\}$, uniqueness implies that the mapping
$(\tau_0,\ldots,\tau_{s-1})\mapsto \sum_{k=0}^{s-1}\tau_k3^k$
is a bijection onto the centered offsets. Therefore, every destination node has exactly one signed base-$3$ path from every source node. It remains to show balance. Fix a phase $k$ and a node $i$. Before phase $k$, a block with digit vector $\tau$ is located at
$(r+\sum_{\ell<k}\tau_\ell3^\ell)\bmod n$.
For every digit vector $\tau$, there is exactly one source $r$ such that this node is $i$, namely
$r\equiv i-\sum_{\ell<k}\tau_\ell3^\ell \pmod n$.
Thus, the blocks stored at $i$ before phase $k$ are in one-to-one correspondence with all $3^s=n$ digit vectors. Among these vectors, exactly $3^{s-1}=\frac{n}{3}$ have $\tau_k=+1$, exactly $\frac{n}{3}$ have $\tau_k=-1$, and exactly $\frac{n}{3}$ have $\tau_k=0$. Hence, in phase $k$, node $i$ sends $\frac{n}{3}$ blocks to the right, $\frac{n}{3}$ blocks to the left, and keeps $\frac{n}{3}$ blocks local.
\end{proof}

\section{Additional Evaluation}

\begin{figure}[H]
    \includegraphics[width=\linewidth]{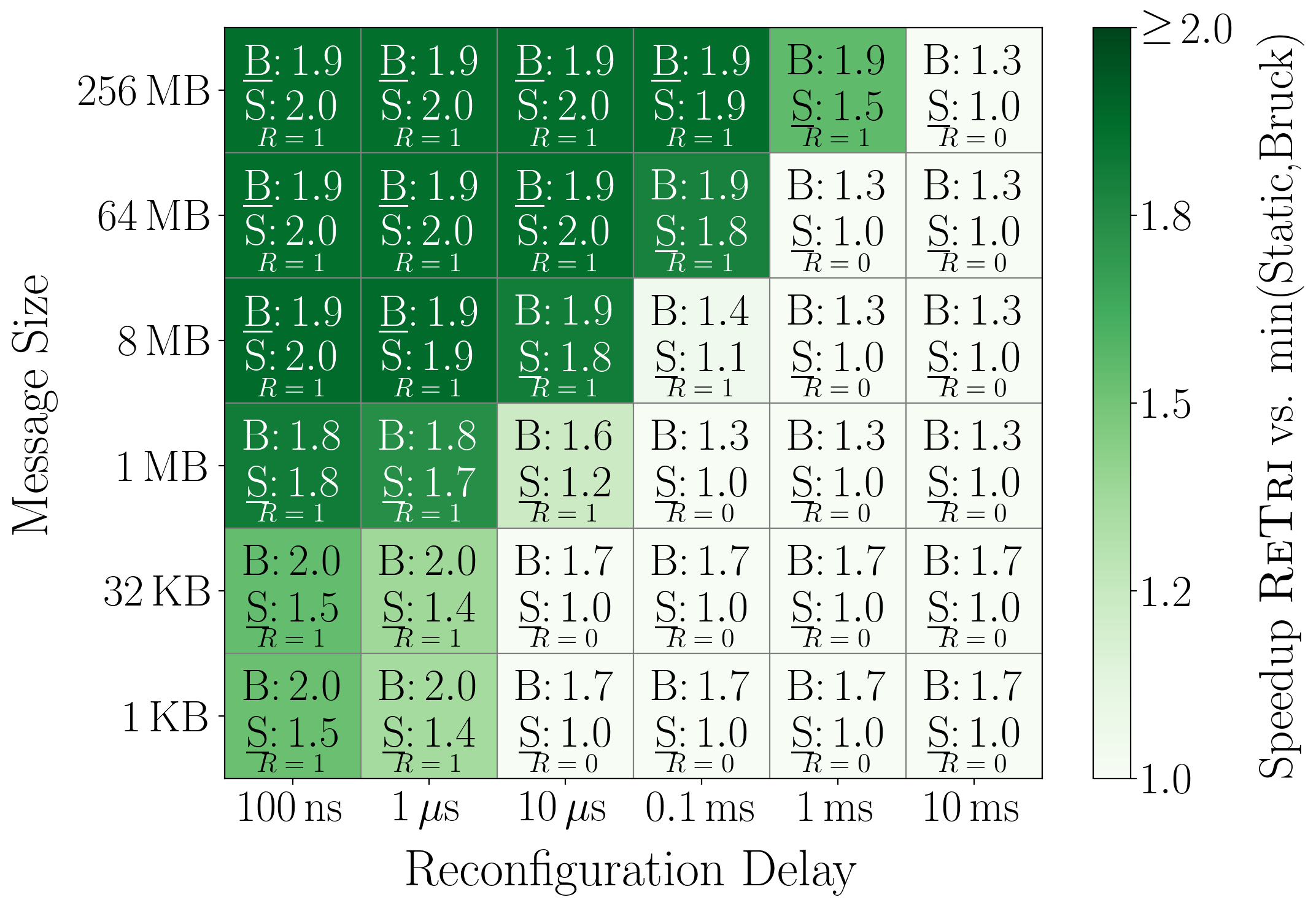}
    \caption{Heatmap showing All-to-All speedups of \name ($n=9$) compared to both reconfigurations with Bruck (B) ($n=8$) and static Pairwise All-to-All (S). The underlined speedup shows the better performing baseline.}
    \label{fig:maxtrix_8_all}
\end{figure}

\begin{figure}[H]
    \includegraphics[width=\linewidth]{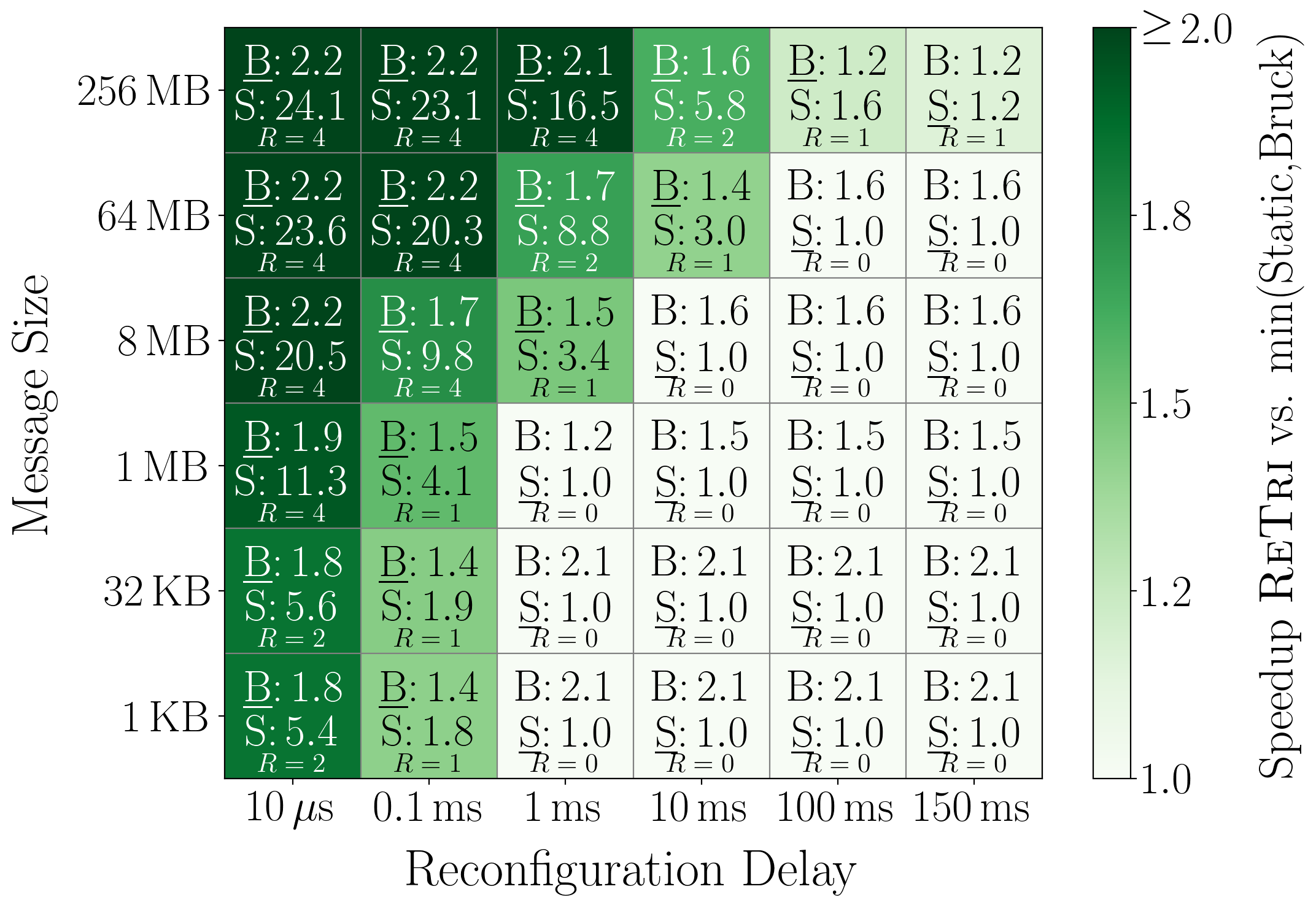}
    \caption{Heatmap showing All-to-All speedups of \name ($n=243$) compared to both reconfigurations with Bruck (B) ($n=256$) and static Pairwise All-to-All (S). Since Bruck operates on a larger network for this experiment, the completion times are normalized by the number of nodes in the network. Each completion time is divided by the number of nodes and only then compared to \name to compute the relative performance.}
    \label{fig:maxtrix_256_all}
\end{figure}
\end{document}